\documentclass[12pt]{article}

\usepackage{amssymb}
\usepackage{amsthm}
\usepackage{amsmath} 
\usepackage{color}
\usepackage[all,web]{xy}
\usepackage{float}
\usepackage{datetime}
\usepackage{pifont}

\usepackage{tikz}
\usetikzlibrary{arrows}
 \usepackage{tensor}
 \usepackage{stmaryrd,bm}
 \usepackage{soul}
 \usepackage{cancel}
\usepackage{lineno}
\usepackage{mathrsfs}
\usepackage{cite}

\newcommand{\ep}{\epsilon}
\newcommand{\ei}{\varepsilon}
\newcommand{\bep}{\bm \ei}
\newcommand{\epo}{\widetilde{\ei}}
\newcommand{\bepo}{\widetilde{\bep}}

\newcommand{\tb}{\bm \tau}

\newcommand{\BH}{\bm H}
\newcommand{\om}{\omega}
\newcommand{\br}{\bm r}

\newcommand{\de}{\partial}
\newcommand{\te}{\theta}
\newcommand{\Be}{\bm q}
\newcommand{\BE}{\bm E}
\newcommand{\BD}{\bm D}
\newcommand{\BI}{\bm I}

\newcommand{\la}{\left \langle}
\newcommand{\ra}{\right \rangle}
\newcommand{\BEA}{\la {\BE} \ra}
\newcommand{\BDA}{\la {\BD} \ra}
\newcommand{\bet}{\hat{\bm q}}

\newcommand{\ui}{u^{in}}
\newcommand{\ue}{u^{ex}}

\newcommand{\ut}{\widetilde{u}}
\newcommand{\vt}{\widetilde{v}}
\newcommand{\re}{\mbox{Re\,}}
\newcommand{\im}{\mbox{Im\,}}
\newcommand{\rf}[1]{(\ref{#1})}
\newcommand{\dst}{\displaystyle}
\newcommand{\zb}{\bar{z}}
\renewcommand{\P}{P_{m,n}}
\newcommand{\nonu}{\nonumber}

\newcommand{\np}{\nabla^\perp}
\renewcommand{\leq}{\leqslant}
\renewcommand{\geq}{\geqslant}
\newcommand{\s}{\mathcal S}

\newcommand{\si}{\s_{in}}
\newcommand{\se}{\s_{ex}}
\newcommand{\ds}{\de \s}
\newcommand{\dsi}{\de \si}
\newcommand{\dse}{\de \se}
\newcommand{\etat}{\tilde{\eta}_2}
\DeclareMathSymbol{\Zeta}{\mathalpha}{operators}{"5A}

\renewcommand{\theequation}{\thesection.\arabic{equation}}
\addtocounter{figure}{0}

\newtheorem{theorem}{Theorem}

\newtheorem{lemma}{Lemma}

\theoremstyle{remark}

\newcommand{\I}{\ensuremath{\mathrm{i}}}
\newcommand{\E}{\ensuremath{\mathrm{e}}}
\newcommand{\D}{\ensuremath{\mathrm{d}}}

\begin{document}

\title{Dispersive and effective properties of two-dimensional periodic media\\}
\author{Yuri A. Godin\thanks{Department of Mathematics and Statistics, University of North Carolina at Charlotte, 
Charlotte, NC 28223 USA. E-mail: ygodin@uncc.edu}~ and Boris Vainberg\thanks{Department of Mathematics and Statistics, 
University of North Carolina at Charlotte, 
Charlotte, NC 28223 USA. Email: brvainbe@uncc.edu}}

\date{}
\maketitle

\begin{abstract}
We consider transverse propagation of electromagnetic waves through a two-dimensional composite material containing
a periodic rectangular array of circular cylinders. Propagation of waves is described by the Helmholtz equation with
the continuity conditions for the tangential components of the electric and magnetic fields on the boundaries of the cylinders.
We assume that the cell size is small compared to the wavelength, but large compared to the radius $a$ of the inclusions.
Explicit formulas are obtained for asymptotic expansion of the solution of the problem in terms of the dimensionless magnitude $q$ of the wave vector
and radius $a$. This leads to explicit formulas for the effective dielectric tensor and the dispersion relation with the rigorously
justified error of order $O((q^2+a^2)^{5/2})$.
\end{abstract}

\section{Introduction}
\setcounter{equation}{0}
Periodic media have attracted a great deal of attention due to the possibility of manipulating the dispersion relation.
In the case of electromagnetic waves, such media known as photonic crystals \cite{JJWM:11} exhibit strong anisotropy of wave
propagation including its total suppression \cite{FGV:98, Kosaka:99}, nonreciprocal wave transmission \cite{FV:01},
slow light \cite{FV:06,MV:04,Kraus:07}, superlensing, \cite{LJJP:03} and more. The advent of metamaterials
has allowed for the engineering of new tunable and switchable devices on the length scale \cite{ZK:12}.

In this paper we study the propagation of waves in a doubly periodic array of scatterers. The multipole expansion method introduce in \cite{R:92}
was applied to the propagation of electromagnetic waves in a doubly periodic lattice in \cite{McPhedran:1996}, while in \cite{Zalipaev:02}
this approach was employed in the problem of elastic wave propagation in a two-dimensional solid containing a doubly periodic array
of circular holes.
Using the  method of matched asymptotic expansion, a dispersion relation was obtained
in \cite{McIver:07} for a doubly periodic
array of small rigid scatterers and in \cite{McIver:2011} for elastic waves in a lattice of cylindrical cavities. Application
of the method to the scatterers with homogeneous Dirichlet boundary conditions was considered in \cite{McIver:2009}
and \cite{Craster:2017}.
A rigorous analysis of a sub-wavelength plasmonic crystal was presented in \cite{Lipton:2010}, where solution of
a nonlinear eigenvalue problem is given in terms of convergent high-contrast power series for the electromagnetic fields and the first branch
of the dispersion relation.

We consider transverse propagation of electromagnetic waves through a two-dimensional composite material
containing a periodic rectangular array $\mathcal C$ of circular cylinders with a positive finite dielectric constant $\ei$. The periods of the lattice
$\tb_1$ and $\tb_2$ are normalized in such a way that $\ell = \min \{|\tb_1|, |\tb_2|\}=1$, while the radius
of the cylinders $a < 0.5$  (see Figure 1). We assume that the relative magnetic permeability of the cylinders
and the matrix equals unity.
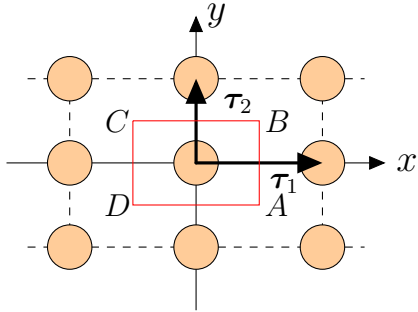
\begin{figure}[H]
\begin{center}
\begin{tikzpicture}[>=triangle 45,scale=0.56]

\draw [->] (-4.5,0) -- (4.5,0)  node[right] {\large $x$};
\draw [->] (0,-3.5) -- (0,3.5)  node[right] {\large $y$};

\draw [dashed] (-3,-2) -- (-3,2);
\draw [dashed] (3,-2) -- (3,2);

\draw [dashed] (-4.0,2) -- (4.0,2);
\draw [dashed] (-4.0,-2) -- (4.0,-2);

\foreach \x in {-3,0,3}
        \draw[fill=orange!40, opacity=1.0,thin] (\x,0)  circle (15pt);
\foreach \x in {-3,0,3}
         \draw[fill=orange!40, opacity=1.0,thin] (\x ,2)  circle (15pt);
\foreach \x in {-3,0,3}
         \draw[fill=orange!40, opacity=1.0,thin] (\x ,-2)  circle (15pt);
\draw [->,very thick,cap=round] (0,0) -- (0,2);
\node [below] at (1.0,1.90) { $\bm{\tau}_2$};
\draw [->,very thick,cap=round] (0,0) -- (3,0);
\node [below] at (2.1,-0.00) {$\bm{\tau}_1$};

\draw [red] ((-1.5,1) --(1.5,1) -- (1.5,-1) -- (-1.5,-1) -- (-1.5,1);

\node [right] at (1.5,-1) {$\!A$};
\node [right] at (1.5,1) {$\!B$};
\node  at (-1.5,1) {$C\quad$};
\node  at (-1.5,-1) {$D\quad$};

\end{tikzpicture}
\label{fig:1}
\caption{Geometry of rectangular lattice of cylinders and the fundamental cell $ABCD$.}
\end{center}
\end{figure}

In dimensionless variables, propagation of the TE mode $\BH = (0,0,u)$ in the $xy$-plane
 is described by the equation
\begin{equation}
 -\nabla \cdot \left(\ep^{-1} \nabla u(\br)\right) = \nu^2\, u(\br), \quad \br\notin \partial{\cal C}
 \label{Hz}
\end{equation}
where
\begin{equation}
 u = \left\{
 \begin{array}{cc}
  u_{in}, & \br \in {\cal C}, \\[2mm]
  u_{ex}, & \br \notin {\cal C},
 \end{array}
\right. \quad
\ep = \left\{
 \begin{array}{ll}
  \ei, & \br \in {\cal C}, \\[2mm]
  1, & \br \notin {\cal C},
 \end{array}
\right.
\end{equation}
${\br} = (x,y)$, $\nu = \dfrac{\om}{c} \ll 1$, where $c$ is the speed of light in vacuum, $\omega$ is the
frequency of the incident wave and $\nu$ is normalized by the condition $\ell=1
$. On the cylinders boundary $\partial\mathcal C$ we impose
continuity conditions of the tangential components of $\BH (\br )$ and $\BE (\br )$
\begin{align}
\label{bc1}
\left. \left\llbracket u (\br) \right\rrbracket \right . &=0, \\[2mm]
\left. \left\llbracket \frac{1}{\ep} \frac{\de u (\br)}{\de n} \right\rrbracket \right. &=0.
\label{bc2}
\end{align}
Hereafter, brackets $\llbracket \cdot \rrbracket$ denote the jump  of the enclosed quantity across
the interface of the cylinders. In addition, $u(\br)$ must satisfy the Floquet-Bloch condition
\begin{equation}
 u(\br + \tb) = \E^{\I\Be \cdot \tb} u(\br),
\label{FB}
\end{equation}
where $\tb$ is any of the lattice periods, $\Be=(q_x, q_y)= q \bet$ runs the primitive cell of the dual lattice
 with $\bet = (\cos \te, \sin \te)$ being the unit vector. This condition implies that the function
$\E^{-\I\Be \cdot \br} u(\br)$ is periodic over the fundamental cell $ABCD$ that we symbolically write as
\begin{equation}
 \rrbracket  \E^{-\I\Be \cdot \br} u(\br) \llbracket =0.
\label{FB1}
\end{equation}
Subsequently, the inverted brackets $\rrbracket  \cdot \llbracket$ denote the jump of the enclosed expression and their
first derivatives across the opposite sides of the cells of periodicity.

The outline of the paper is the following.
In Section 2 we formulate main results. In Section 3
we seek the solution of the problem as a power series in terms of the absolute value of the quasimomentum $q$, and then derive
recurrence relations between {\bf the coefficients $u_k$ of the power series. We prove} that the coefficients are odd
functions in space variables if $k$ is odd, and otherwise even. We also prove that the power series for the eigenvalues
contains only even powers of $q$. In Section 4 we obtain explicit formulas for the coefficients of the series with
a given accuracy in terms of the radius $a$ of the cylinders. Explicit approximations of the effective tensor with the
accuracy $O((q^2+a^2)^{5/2})$ and the dispersion relation with the accuracy $O(q^2(q^2+a^2)^{2})$ are obtained in
Section 5. To prove these results we show in Appendix that the power series in $q$ of the solution $u$ and of the
eigenfrequency converge uniformly in $a$.

We assume that the dielectric constant $\ep$ is positive and fixed (does not depend on $a$ and $q$) while $a$ is small. The case
of lossy composites and metamaterials will be considered elsewhere. Let us stress again that our goal is explicit
formulas with high accuracy and rigorous estimates of the remainders in a multidimensional setting. There are many
papers where similar problems were often solved under more general assumptions, but with less demanding goals.
See, for example, \cite{Zalipaev:02,Meng:18,Craster:10,Vanel:17,Cherednichenko:06,Smyshlyaev:08,Joyce:17}.

\section{Formulation of the problem and the main result}
\setcounter{equation}{0}

We reduce the above problem to the fundamental cell $\cal S$ centered at the origin:
\begin{equation}
 -\frac{1}{\ep}\,\Delta u = \nu^2 u, \quad \br\in\mathcal S, \quad r=|\br|\neq a,
\label{Hz1}
\end{equation}
\begin{align}
\label{bc1a}
\left\llbracket u (\br) \right\rrbracket =0, \quad
\left\llbracket \frac{1}{\ep} \frac{\de u (\br)}{\de n} \right\rrbracket =0, \quad
\rrbracket  \E^{-\I\Be \cdot \br} u(\br) \llbracket =0 ~~{\rm on} ~~ \partial\mathcal S.
\end{align}

The main result of the paper concerns approximation of the effective dielectric tensor $\bep^\ast$ defined
from $\BDA = \bep^\ast \BEA$, where
$\dst \BEA = \frac{\I c}{\omega}\,\la \frac{1}{\ep} \nabla \times (0,0,u) \ra$
is the average electric field and $\BDA = \la \ep \BE \ra$
is the average electric displacement. Here $u$ is the $z$-component of the magnetic field: $\BH=(0,0,u)$.
It states that in the low frequency regime with small inclusions when $q^2 + a^2 \ll 1,$ we have
\begin{align}
\label{bep}
 \bep^\ast &= \left(1 + \frac{2\pi \alpha a^2}{\tau_1 \tau_2}
 +\frac{1}{12}\, \frac{\pi \alpha a^2 q^2}{\tau_1 \tau_2} \left( \tau_1^2 \cos^2 \te
 + \tau_2^2 \sin^2 \te \right) \right) \BI \nonu \\[2mm]
 &+ \frac{4\pi \alpha^2 a^4}{\tau_1^2 \tau_2^2 } \left[
\begin{array}{cc}
 \eta_1 \tau_2 & 0 \\[1mm]
 0 & \etat \tau_1
\end{array}
\right]+ O\left ( \left(q^2 + a^2\right)^{\frac{5}{2}} \right),
\end{align}
where $\alpha = \dfrac{\ei - 1}{\ei + 1}$, $\tau_k = |\tb_k|, k=1,2$, $\eta_1 = \zeta (\tau_1/2)$, $\etat = \I \zeta \left( \I \tau_2/2 \right)$,
and $\zeta (z)$ is the Weierstrass zeta-function \cite{Whittaker:1927}. We also obtain an approximation of the dispersion relation
\begin{align}
 \nu^2 = q^2 \left(1 - \frac{2\pi\alpha a^2}{\tau_1 \tau_2} \right) + O\left( q^2 (q^2 + a^2)^2 \right).
\end{align}
The method used in the paper can be used to obtain the above expressions with higher accuracy.

\section{Series expansion of the field}
\setcounter{equation}{0}

Elliptic problem \rf{Hz1}-\rf{bc1a} is symmetric, depends analytically on $q$, and has a simple
eigenvalue $\nu^2 = 0$ when $q=0$ with the eigenfunction $u=const$. Thus the eigenvalue $\nu^2$  depends analytically
on $q$ for $q \ll 1$, and the eigenfunction $u(\br,\Be)$ can be chosen to be analytic in $q$, i.e., for small $q$
we can expand $u$ and $\nu^2$  in a power series
\begin{align}
 \label{u_ser}
 u(\br, \Be) &=  1 +qu_1(\br, \bet) + q^2 u_2(\br, \bet) + q^3u_3(\br, \bet) + \ldots, \\[2mm]
\nu^2 &= q\lambda_1 + q^2\lambda_2 + q^3\lambda_3 + \ldots.
 \label{nu_ser}
\end{align}
The latter series can be viewed as a perturbation of a simple eigenvalue $\nu^2 = 0$ corresponding
to the eigenfunction $u = 1$. The rigorous justification of \rf{u_ser}, \rf{nu_ser} will be given in the Appendix.
It will be shown there that series \rf{u_ser} converges in the Sobolev space $H^1(S),$ and both of them are uniform in $a$.
Moreover, it will be shown below that series \rf{nu_ser} contains only even powers of $q$, i.e., $\lambda_{2n+1}=0$.

Substituting expansions \rf{u_ser}--\rf{nu_ser} into \rf{Hz1} we obtain a system of recurrence equations
for determination of $u_n$
\begin{align}
 \label{u_1}
 -\frac{1}{\ep}\, \Delta u_1 &= \lambda_1, \quad r\neq a\\[2mm]
 \label{u_2}
 -\frac{1}{\ep}\, \Delta u_2 &= \lambda_2 + \lambda_1 u_1, \quad r\neq a, \\[2mm]
 \label{u_3}
 -\frac{1}{\ep}\, \Delta u_3 &= \lambda_3 + \lambda_2 u_1 + \lambda_1 u_2, \quad r\neq a, \\
  & ~\,\vdots& \nonu \\[-2mm]
 \label{u_k}
 -\frac{1}{\ep}\, \Delta u_k &= \lambda_k + \sum_{n=1}^{k-1} \lambda_{k-n} u_n, \quad r\neq a, \quad k\geq 4.
\end{align}
On the boundary $r=a$ functions $u_k$ satisfy the conditions (in what follows we omit dependence of $u_k$ on $\bet$ 
for brevity)
\begin{align}
\label{bc1uk}
\left. \left\llbracket u_k (\br) \right\rrbracket \right . &=0, \\[2mm]
\left. \left\llbracket \frac{1}{\ep} \frac{\de u_k (\br)}{\de n} \right\rrbracket \right. &=0,
\label{bc2uk}
\end{align}
while on $\partial \cal S$ we have a system of recurrence equations
\begin{align}
 \label{u1b}
 \rrbracket u_1(\br) \llbracket &=  \rrbracket \I\bet \cdot \br \llbracket, \\
 \label{u2b}
 \rrbracket u_2(\br) \llbracket &=  \rrbracket (\I\bet \cdot \br)u_1 - \frac{1}{2!}\,(\I\bet \cdot \br)^2\llbracket, \\
 \label{u3b}
 \rrbracket u_3(\br) \llbracket &=  \rrbracket (\I\bet \cdot \br)u_2 - \frac{1}{2!}\,(\I\bet \cdot \br)^2 u_1
 + \frac{1}{3!}\,(\I\bet \cdot \br)^3 \llbracket, \\
 & ~\,\vdots& \nonu \\
 \label{ukb}
 \rrbracket u_k(\br) \llbracket &= \left\rrbracket \frac{(-1)^{k+1}}{k!}\,\left(\I \bet \cdot \br \right)^k
 + \sum_{n=1}^{k-1} \frac{(-1)^{n+1}}{n!}\, \left(\I \bet \cdot \br \right)^n u_{k-n} \right\llbracket, \quad k \geq 4.
\end{align}

In what follows we need to establish some important properties of the functions $u_k$. Firstly, we normalize
$u(\br, \Be)$ in such a way that $\int_{S}u(\br, \Be)\,dS=|\s|=\tau_1\tau_2$. This implies that
\begin{equation}
\int_{\s}u_k(\br)dS=0.
 \label{intzero}
\end{equation}

We also will need Green's formula for solutions of \rf{Hz1}, \rf{bc1a}:
\begin{equation}
 \int_{\s} \frac{1}{\ep}\, |\nabla u|^2\,\D S = \nu^2 \int_{\mathcal S} |u|^2\,\D S.
 \label{green}
\end{equation}
that follows from the symmetry of the problem (\ref{Hz1}), (\ref{bc1a}).
Indeed, let  $\s = \si \cup \se$ where $\si$ is the disk $r<a$. One can multiply both sides of \rf{Hz1} by the complex conjugate $\bar{u}$ of $u$ and
apply Green's first identity to each part of $\mathcal S$.
When we add up the identities, the contour integrals over the boundary
$r=a$ are cancelled due to \rf{bc1a}, and \rf{green} follows.

Consider an auxiliary problem for the function $v(\br)$
\begin{equation}
 -\frac{1}{\ep} \,\Delta v = f, \quad \br \in \s, \quad r\neq a,
 \label{eqv}
\end{equation}
with the homogeneous conditions
\begin{equation}
 \left. \left\llbracket v (\br) \right\rrbracket \right . =0, \quad
\left. \left\llbracket \frac{1}{\ep} \frac{\de v (\br)}{\de n} \right\rrbracket \right. =0, \quad
\rrbracket v(\br) \llbracket = 0.
\label{conv}
\end{equation}

\begin{lemma}
  \label{lemma1}
  (a) Problem \rf{eqv}-\rf{conv} with $f=0$ has a unique solution $v=const$.
  (b) The nonhomogeneous problem \rf{eqv}-\rf{conv} has a solution if and only if $f$ is orthogonal to a constant.
\end{lemma}
\begin{proof}
 First statement follows from the application of Green's formula \rf{green} with $\nu=0$ to \rf{eqv}.
 The second statement is the Fredholm alternative applied to equation \rf{eqv}.
\end{proof}

\begin{lemma}
 \label{lemma2}
 The pair $u_k(\br), \lambda_k$ is defined uniquely from \rf{u_1}-\rf{intzero}, i.e., problem \rf{u_1}-\rf{intzero} does not have solutions that are different from those defined in \rf{u_ser}, \rf{nu_ser}.
\end{lemma}
\begin{proof}
 Let $k \geq 1$ be the least number for which there are two different pairs $u_k(\br), \lambda_k$.
 Then 3.3-3.6 implies that there are two different functions $u^{(1)}_k (\br)$ and $u^{(2)}_k (\br)$.
 Hence their difference $v (\br) = u^{(1)}_k (\br) - u^{(2)}_k (\br)$ satisfies \rf{eqv}-\rf{conv}
 with $f = \lambda_k^{(1)} - \lambda_k^{(2)}$. From the previous lemma it follows that
 $\lambda_k^{(1)} = \lambda_k^{(2)}$ and $v (\br) = const$. The latter together with \rf{intzero} implies
 $v (\br) =0$.
\end{proof}

We will use the term odd or even function if the corresponding property holds with respect to the origin, i.e.,
a scalar function $f(\br)$ is odd if $f(-\br)=-f(\br)$ and is even if $f(-\br)=f(\br)$.
Now we can formulate the result concerning the structure of expansions \rf{u_ser}-\rf{nu_ser}.
\begin{theorem}
 \label{theorem1}
 Functions $u_k(\br, \bet)$ in expansion \rf{u_ser} are odd functions of $\br$ for odd $k$
 and even ones if $k$ is even. Expansion \rf{nu_ser} of $\nu^2$ contains only even powers of $q$,
 i.e. $\lambda_{2k-1} = 0, \; k=1,2,\ldots$
\end{theorem}
\begin{proof}
 We prove the theorem by induction in $k$. For $k=1$ the boundary condition \rf{u1b} is odd. Then the even component
 $v (\br)$ of $u_1$ is the solution of \rf{eqv}-\rf{conv} with $f=\lambda_1$. From lemma \ref{lemma1} it follows that
 $\lambda_1=0$ and lemma \ref{lemma2} implies that $v=0$. Hence, the statement of the theorem is valid for $k=1$.
 Assume now that the statement of the theorem holds for $1 \leq k < k_0$.
 Let us prove it for $k=k_0$. We need to consider two cases of even and odd $k_0$.
 \begin{itemize}
 \item[]
  Case 1: If $k_0 = 2m$ then $\lambda_{k-n}=0$ in \rf{u_k} when $n$ is odd. Thus, the right-hand side
  of \rf{u_k} is even by the induction hypothesis. The right-hand side of \rf{ukb} is also even. Thus,
  the odd component of $u_{2m}$ satisfies the homogeneous problem and equals zero due to Lemma \ref{lemma1}.
 \item[]
  Case 2: Let $k_0 = 2m+1$. Then the right-hand side of \rf{u_k} is the sum of $\lambda_{2m+1}$ and an odd function
  by the induction hypothesis. The right-hand side of \rf{ukb} is odd. Thus, the even component of $u_{2m+1}$
  is the solution of \rf{eqv}-\rf{conv} with $f=const = \lambda_{2m+1}$. From Lemmas \ref{lemma1} and \ref{lemma2}
  it follow that $\lambda_{2m+1}=0$ and the even component of $u_{2m+1}$ is zero.
 \end{itemize}
\end{proof}

Substitution of expansion \rf{u_ser} into \rf{green} and taking into account the oddness and evenness
of $u_k$ leads to approximation of $\nu^2$. In particular, we obtain to the order $O\left( q^6 \right)$
\begin{align}
 \nu^2 = q^2 \, \frac{\dst \int_{\s} \frac{1}{\ep} \left(|\nabla u_1 |^2 + q^2 \left( |\nabla u_2|^2
 +2\re \left( \nabla u_1 \cdot \nabla \bar{u}_3 \right)\right)
  \right)\, \D S }{\dst \int_{\s} \left(1 + q^2 \left( |u_1|^2 +2\re u_2 \right) \right)\, \D S}+ O\left( q^6\right),
 \label{nu2}
\end{align}
where $\bar{u}$ is the complex conjugate of $u$.

\section{A priori estimates for the power series terms}
\setcounter{equation}{0}

Functions $u_k$ in \rf{nu2} and in formulas \rf{epo1}, \rf{epo2} (which are used to find $\varepsilon^*$, see below) are obtained as solutions of certain boundary value problems which depend on $a$ and can be expanded
in power series in $a$. We need some a priori estimates for the solutions of these problems in order to justify
the asymptotic convergence of the power series in $a$. We will start with recalling the Poincar\'{e} lemma, which is so simple in our setting ($\s$ is a rectangle) that we will prove it.
\begin{lemma}
\label{lemma0}
Let $v\in H^1(\s)$ and
\begin{equation}
\int_{\s}v(\br)\, \D S=0.
 \label{intzero1}
\end{equation}
Then $\|v\|_{L_2}\leq C_1 \|\nabla v\|_{L_2}$ and $\|v\|_{H^1}\leq C_2 \|\nabla v\|_{L_2}$.
\end{lemma}
\begin{proof}
We will prove the first inequality since it obviously implies the second with $\dst C_2 = \sqrt{C_1^2 + 1}$. In order to prove the first inequality we write $u$ in the form of the Fourier series:
\begin{align}
\nonumber
v={\sum_{m,n}}^\prime v_{mn} \E ^{2\pi \I \left( \frac{mx}{\tau_1} + \frac{ny}{\tau_2} \right)},
 \end{align}
 where the prime indicates that the term $v_{00}$ is omitted. This term is zero due to \rf{intzero1}. It remains to compare the norms expressed through the Fourier coefficients:
 \begin{align}\nonumber
 \| v\|^2_{L_2}=c{\sum_{m,n}}^\prime|v_{m,n} |^2, \quad \| \nabla v\|^2_{L_2}= c{\sum_{m,n}}^\prime |v_{m,n} |^2
  \left[ \left( \frac{m}{\tau_1}\right)^2 + \left( \frac{n}{\tau_2}\right)^2\right], \quad c= \tau_1 \tau_2
 \end{align}
\end{proof}
\begin{lemma}
\label{lemma3}
 Let $v(\br)$ be the solution of the problem
 \begin{equation}
   \frac{1}{\ep}\Delta v(\br) = f(\br), \quad \br \in \s, \quad r\neq a,
  \label{v}
 \end{equation}
 subject to the conditions
\begin{align}
\left. \left\llbracket v (\br) \right\rrbracket \right . =0, \quad
\left. \left\llbracket \frac{1}{\ep} \frac{\de v (\br)}{\de n} \right\rrbracket \right. =0, \quad
 \rrbracket v(\br) \llbracket =  0, \quad
\end{align}
and let condition \rf{intzero1} hold. Then
\begin{equation}
 \left\| v \right\|_{H^1} \leq C \left\|  f\right\|_{L_2}.
\end{equation}
\end{lemma}
\begin{proof}
 Multiplying \rf{v} by $\bar{v}$ and applying Green's first identity we obtain
 \begin{align}
   \int_{\s}  f \bar{v} \, \D S =\int_{\s} \frac{1}{\ep} \bar{v} \Delta v \, \D S
  =\int_{\s} \frac{1}{\ep} \left| \nabla v \right|^2 \D S = \left\| \frac{1}{\ep}\nabla v \right\|^2_{L_2}.
  \label{green2}
 \end{align}
 To be more accurate, one needs to write Green's first identities separately for each part $S_{in}, S_{ex}$ of $S$, add them and check that the contour integrals over the boundary $r=a$ are cancelled. Equality \rf{green2} implies $\left\| \nabla v \right\|^2_{L_2}\leq \left\| f \right\|_{L_2}\left\| v \right\|_{L_2}$. It remains to apply Lemma \ref{lemma0}.

\end{proof}
Next lemma shows that a similar estimate holds when an inhomogeneity appears in the boundary condition.
\begin{lemma}
 \label{lemma4}
 Suppose that $v(\br)$ satisfies
 \begin{equation}
  \Delta v(\br) = 0, \quad \br \in \s, \quad r\neq a,
 \end{equation}
 the boundary conditions
\begin{align}
\left. \left\llbracket v (\br) \right\rrbracket \right . =h(\phi), \quad
\left. \left\llbracket \frac{1}{\ep} \frac{\de v (\br)}{\de n} \right\rrbracket \right. = g(\phi), \quad
 \rrbracket v(\br) \llbracket =  0,
\label{v_cond}
\end{align}
where $\phi$ is the polar angle. Let also $ v(0) =0$. Then
\begin{equation}
 \left\| \nabla v \right\|_{L_2} \leq c_1 \left( \int_{0}^{2\pi} \left(g^2(\phi)
 + \left[g^{\prime\prime}(\phi)\right]^2 \right)\D \phi \right)^{1/2} +
 \frac{c_2}{a} \left( \int_{0}^{2\pi} \left(h^2(\phi)
 + \left[h^{\prime\prime}(\phi)\right]^2 \right)\D \phi \right)^{1/2}.
\end{equation}
\end{lemma}
\begin{proof}
 Let $h=0$.
 After the substitution
 \begin{equation}
 \label{v1}
 v = v_1 + w, \quad \text{where} \quad
 v_1 = \left\{
 \begin{array}{cl}
 \dfrac{\ei r^2}{a^2} (a-r)g(\phi), & r<a, \\[2mm]
 0, & r>a,
 \end{array}
 \right.
 \end{equation}
 the problem for $w$ is reduced to that outlined in the previous lemma with
 \begin{equation}
 f = \left\{
 \begin{array}{cl}
 \dfrac{\ei}{a^2} \left( (9r-4a) g(\phi) + (r-a) g^{\prime\prime}(\phi) \right), & r<a, \\[2mm]
 0, & r>a.
 \end{array}
 \right.
 \end{equation}
 If there is a jump $h(\phi)$ of the function in \rf{v_cond} instead of the jump $g(\phi)$ of the derivative
 then function $v_1$ in \rf{v1} must be replaced by the function
  \begin{equation}
 v_1 = \left\{
 \begin{array}{cl}
 \dfrac{r^2}{a^3} \left( 2r -3a \right)h(\phi), & r<a, \\[2mm]
 0, & r>a.
 \end{array}
 \right.
 \end{equation}
\end{proof}
If an inhomogeneity is present in both the equation and the boundary conditions then the sum of the estimates from
lemma \ref{lemma3} and lemma \ref{lemma4} gives an estimate of the norm of the gradient.

\subsection{Approximation of $u_1$}

Harmonic function $u_1$ is a solution of the static (with $q=0$) problem.
It is expedient to look for $u_1$ in the form of a power series in the inclusion and a combination of a linear and Weierstrass' zeta-function and its derivatives outside the inclusion.
The choice of the Weierstrass function is dictated by periodicity of the composite and the periodic properties of the Weierstrass function. This approach was used in
\cite{G:13}. Here we use similar representation of $u_1$ and achieve desired accuracy in $a$ using a finite number of terms of the corresponding series.

We introduce complex variable $z=x+\I y=r\E^{\I\phi}$, and along with vector periods $\tb_1$ and
$\tb_2$ we will use their complex counterparts $\tau_1>0$ and $\I \tau_2$, $\tau_2 >0$.
Weierstrass' zeta-function \cite{Whittaker:1927} is defined by
\begin{equation}
 \zeta (z) = \frac{1}{z} + {\sum_{m,n}}^\prime \left[ \frac{1}{z-\P} + \frac{1}{\P}
+ \frac{z}{\P^{\,2}} \right],
\label{zeta}
\end{equation}
where  $\P = m\tau_1 + \I n\tau_2$ are coordinates of the lattice nodes in the complex plane.
Prime in the sum means that summation is extended over all pairs $m,\,n$ except $m=n=0$.
We use its quasiperiodicity property
\begin{align}
 \label{quasi1}
 \zeta(z+\tau_1) - \zeta(z) &= 2\eta_1, \hspace{1ex} \eta_1 = \zeta\left(\tau_1/2 \right), \\[2mm]
 \zeta(z+\I\tau_2) - \zeta(z) &= 2\eta_2, \hspace{1ex} \eta_2 = \zeta\left(\I\tau_2/2 \right),
 \label{quasi2}
 \end{align}
where for rectangular lattices $\eta_1$ is purely real while $\eta_2$ is purely imaginary. It is convenient
to introduce real parameter $\etat = \I \eta_2$. If we subtract from $\zeta (z)$ its linear part then
the resulting function will be periodic and harmonic. Thus,
\begin{equation}
\label{period}
 \left\rrbracket \zeta(z) - \frac{2\eta_1}{\tau_1}\,x + \frac{2\I \etat}{\tau_2}\, y \right\llbracket =0.
\end{equation}
This property is used in the lemma below to find an approximation $\ut_1$ to $u_1$ to the order
$\dst O\left(a^5 \right)$.
\begin{lemma}\label{appu1}
Denote $\ut_1^{in}=\ut_1, r<a,$ and $\ut_1^{ex}=\ut_1, r>a.$ Let
\begin{align}
 \label{u1_in_appr}
 \ut_1^{in} & = \I r(A_1 \cos \phi + B_1 \sin \phi), \\[2mm]
 \ut_1^{ex} & =\I\bet \cdot \br + \I a^2\re\left[(C_1 + \I D_1) \left(\zeta(z)-\frac{2\eta_1}{\tau_1}\,x
 +\frac{2\I \etat}{\tau_2}\,y\right)\right].
 \label{u1_ex_appr}
 \end{align}
 where real constants $A_1,B_1, C_1, D_1$ are given below. Then $\|u_1-\ut_1\|_{H^1}\leq Ca^5$.
 \end{lemma}
 \begin{proof}
Let us substitute $\ut_1$ into \rf{u_1}, \rf{bc1uk}-\rf{bc2uk} (with $k=1$) and  \rf{u1b}.
Functions \rf{u1_in_appr}, \rf{u1_ex_appr} are harmonic, i.e., \rf{u_1} holds for $\ut_1$.
Due to \rf{period} property  \rf{u1b} is satisfied for $\ut_1^{ex}$ and its normal derivatives.

To satisfy conditions \rf{bc1uk}-\rf{bc2uk} on the boundary $r=a$ we expand $\zeta(z)$
in a Laurent series
\begin{align}
\zeta (z) = \frac{1}{z} - \sum_{k=2}^\infty s_{2k} z^{2k-1},
 \label{zeta_ser}
\end{align}
where $s_{2k}$ are real lattice sums
\begin{equation}
 s_{2k} = {\sum_{m,n}}^\prime \frac{1}{\P^{2k}}, \quad k=2,3,\ldots.
 \label{s}
\end{equation}
We substitute \rf{zeta_ser} into \rf{u1_ex_appr} and equate the coefficients of $\cos \phi$ and
$\sin\phi$ in \rf{bc1uk}-\rf{bc2uk}.  This leads to
\begin{alignat}{4}
   \label{A0}
  &A_1 &&= \frac{2\ei}{\ei +1}\left(1 + \frac{2\alpha \eta_1}{\tau_1}\,a^2 \right)^{-1} \cos \te, \quad
  && C_1 &&= \alpha \left(1 + \frac{2\alpha \eta_1}{\tau_1}\,a^2 \right)^{-1} \cos \te, \\[2mm]
  &B_1 &&= \frac{2\ei}{\ei +1}\left(1 + \frac{2\alpha \etat}{\tau_2}\,a^2 \right)^{-1} \sin \te, \quad
  && D_1 &&= \alpha \left(1 + \frac{2\alpha \etat}{\tau_2}\,a^2 \right)^{-1} \sin \te,
  \label{D0}
\end{alignat}
where $\alpha = \dfrac{\ei -1}{\ei + 1}$.
Hence, approximation \rf{u1_in_appr}-\rf{u1_ex_appr} satisfies exactly \rf{u_1} and \rf{u1b}. Conditions \rf{bc1uk}-\rf{bc2uk}
are satisfied exactly only for the terms containing $\cos \phi$ and $\sin \phi$ but have an error in the terms
$\cos n\phi, \sin n\phi$ with $n \geq 3$. This error is large, and
Lemma \ref{lemma4} does not allow us to justify that \rf{u1_in_appr}-\rf{u1_ex_appr} approximates $u_1$ with the desired accuracy. Therefore we will add
an extra term to $\widetilde{u}_1$, but later it will be shown that this extra term can be omitted. Hence, let $\widetilde{v}_1=\widetilde{v}_1^{in}, r<a,~\widetilde{v}_1=\widetilde{v}_1^{ex}, r>a$, where
\begin{align}
 \label{u1_in_apprV}
\vt_1 ^{in}& = \I r(A_1 \cos \phi + B_1 \sin \phi) + \I r^3(A_2 \cos 3\phi + B_2 \sin 3\phi), \\[2mm]
 \vt_1^{ex}  & =\I\bet \cdot \br + \I a^2\re\left[(C_1 + \I D_1) \left(\zeta(z)-\frac{2\eta_1}{\tau_1}\,x
 +\frac{2\I \etat}{\tau_2}\,y\right)\right]
 + \I a^4 \re \left[(C_2 + \I D_2)\, \zeta^{\prime\prime}(z)\right] \nonu \\[2mm]
 & = \I\bet \cdot \br + \I a^2 \re \left[(C_1 + \I D_1) \left( \frac{1}{z} - s_4 z^3 -\frac{2\eta_1}{\tau_1}\,x
 +\frac{2\I \etat}{\tau_2}\,y + O(z^5)\right)\right] \nonu \\[2mm]
 &+ \I a^4 \re \left[(C_2 + \I D_2)\,\left( \frac{2}{z^3} +O(z) \right) \right],
 \label{u1_ex_apprV}
\end{align}
where $A_1,B_1,C_1,D_1$ remain the same. The function above is still harmonic.
Since derivatives of zeta-function are periodic, their addition to $\ut_1^{ex}$ does not violate
\rf{u1b}. We substitute \rf{u1_in_apprV}, \rf{u1_ex_apprV} into \rf{bc1uk}-\rf{bc2uk} and equate coefficients
of $\cos 3\phi$ and $\sin 3\phi$. This gives
\begin{alignat}{4}
   \label{A2}
  & A_2 &&= -\frac{2 \alpha a^2\ei s_4}{\ei + 1}\, \left(1 + \frac{2\alpha \eta_1}{\tau_1}\,a^2 \right)^{-1} \cos \te, \quad
  && C_2 &&= -\frac{1}{2}\,\alpha^2 a^4 s_4 \left(1 + \frac{2\alpha \eta_1}{\tau_1}\,a^2 \right)^{-1} \cos \te, \\[2mm]
  & B_2 &&= \frac{2 \alpha a^2\ei s_4}{\ei + 1}\, \left(1 + \frac{2\alpha \etat}{\tau_2}\,a^2 \right)^{-1} \sin \te, \quad
  &&D_2 &&= \frac{1}{2}\,\alpha^2 a^4 s_4  \left(1 + \frac{2\alpha \etat}{\tau_2}\,a^2 \right)^{-1} \sin \te.
  \label{D2}
\end{alignat}
From \rf{u1_in_apprV}-\rf{D2} it follows that $\vt_1$ satisfies \rf{bc1uk}, \rf{bc2uk} with the
accuracy $O(a^7)$ and $O(a^6)$, respectively. Thus Lemma \ref{lemma4} implies
that $\|u_1-\widetilde{v}_1\|_{H^1}\leq Ca^7$. One can easily check that
$\|\widetilde{u}_1-\widetilde{v}_1\|_{H^1} =O(a^5)$.

\end{proof}

\subsection{Approximation of $u_2$}

Function $u_2$ is an odd one and does not contribute to the average electric field. However, it appears in \rf{u3b} and in \rf{nu2}. Because
of that we will determine $u_2$ to the order
$O\left( a^2 \right)$. In the equation $u_2$ satisfies
\begin{equation}
 -\frac{1}{\ep}\, \Delta u_2 = \lambda_2, \quad r\neq a,
 \label{u2}
\end{equation}
one must know $\lambda_2$ in order to find $u_2$. We will find it with an  accuracy higher than that for $u_2$
since $\lambda_2$ is involved not only in \rf{u2} but also in the dispersion relation \rf{nu_ser}.
\begin{lemma}\label{lambda2}
The following relation is valid for $\lambda_2$:
\[
\lambda_2=1 - \frac{2\pi\alpha a^2}{\tau_1 \tau_2} + O\left( a^4 \right).
\]
\end{lemma}
\begin{proof}
It follows from \rf{nu_ser} and \rf{nu2} that
\begin{align}
 \lambda_2 = \frac{1}{S}\int_{\s} \frac{1}{\ep}\, \left| \nabla u_1 \right|^2 \D S
 = \frac{1}{S} \left(\frac{1}{\ei}\int_{\si} \left| \nabla \ui_1 \right|^2 \D S
 + \int_{\se} \left| \nabla \ue_1 \right|^2 \D S\right).
 \label{lam2}
\end{align}
Lemma \ref{appu1} yields
\begin{align}
 \frac{1}{\ei}\int_{\si} \left| \nabla \ui_1 \right|^2 \D S &=
 \frac{1}{\ei}\int_{\si} \left| \nabla \ut_1^{in} \right|^2 \D S + O\left( a^7 \right) =
 \frac{1}{\ei}\, \pi a^2 \left(A_1^2 + B_1^2 \right) + O\left( a^7 \right) \nonu \\[2mm]
 & = \frac{4\pi a^2 \ei}{(\ei +1)^2}+ O\left( a^4 \right).
\end{align}
From \rf{u1_ex_appr}, \rf{A0}-\rf{D0} we have
\begin{align}
  \label{gu1e}
  &\int_{\se} \left| \nabla \ue_1 \right|^2  \D S = \int_{\se} \left(
  \Bigl( \cos \te - \frac{2a^2 \eta_1}{\tau_1}\,C_1 \Bigr)^2
  + \Bigl( \sin \te - \frac{2a^2 \etat}{\tau_2}\,D_1 \Bigr)^2 \right. \nonu \\[2mm]
  & +2a^2 \Bigl( \cos \te - \frac{2a^2 \eta_1}{\tau_1}\,C_1 \Bigr) \re \left[ (C_1 + \I D_1) \zeta^\prime (z) \right]
   - 2a^2 \Bigl( \sin \te - \frac{2a^2 \etat}{\tau_2}\,D_1 \Bigr) \im \left[ (C_1 + \I D_1) \zeta^\prime (z) \right] \nonu \\[2mm]
  & +a^4 \left(|C_1|^2 + |D_1|^2 \right) |\zeta^\prime (z)|^2 \biggr)\, \D S + O\left( a^5 \right).
\end{align}
To evaluate the integral containing the derivatives
of zeta-function we use Green's theorem along with the quasiperiodicity properties
\rf{quasi1}-\rf{quasi2}:
\begin{align}
 \label{int_zeta_prime_bar}
 \int_{\se} \zeta^\prime (z)\, \D S = \frac{\I}{2} \oint_{\dse} \zeta(z)\,\D \zb
 = \eta_1 \tau_2 - \etat \tau_1 - \frac{\I}{2} \oint_{\dsi} \zeta(z)\,\D \zb = \eta_1 \tau_2 - \etat \tau_1,
\end{align}
where the last integral vanished due to expansion \rf{zeta_ser}. Finally, we need to evaluate
the integral of $|\zeta^\prime (z)|^2$. The integral of the regular part of $\zeta^\prime (z)$ is bounded in $a$. Therefore we have
\begin{align}
 \int_{\se} |\zeta^\prime (z)|^2 \, \D S &= \int_{\se} \frac{1}{r^4}\, r\D r \D \phi + O\left( 1 \right)
 = \int_{r>a} \int_{0}^{2\pi}\frac{\D r \,\D \phi}{r^3} + O\left( 1\right) = \frac{\pi}{a^2} + O\left( 1\right).
\end{align}
Now using the Legendre relation \cite{Whittaker:1927} which in our case reads
$\eta_1 \tau_2 + \etat \tau_1 = \pi$ we obtain
\begin{align}
 \int_{\se}\left| \nabla \ue_1 \right|^2 \D S &= \tau_1 \tau_2 \left(1 - 4\alpha a^2 \left(\frac{\eta_1}{\tau_1}\cos^2 \te
 + \frac{\etat}{\tau_2} \sin^2 \te \right) \right) - \pi a^2 \nonu \\[2mm]
 &+ 2\alpha a^2 (\eta_1 \tau_2 -\etat \tau_1) (\cos^2 \te - \sin^2 \te) + \pi \alpha^2 a^2 + O\left( a^4 \right) \nonu \\[2mm]
 &= \tau_1 \tau_2 - 2\pi\alpha a^2 -\pi a^2 + \pi \alpha^2 a^2 + O\left( a^4 \right).
\end{align}
Finally substituting all terms in \rf{lam2} we have
\begin{align}
 \label{lambda_2}
 \lambda_2 &= 1 - \frac{2\pi\alpha a^2}{\tau_1\tau_2} + \frac{1}{\tau_1 \tau_2}
 \left( \frac{4\pi a^2 \ei}{(\ei +1)^2}- \pi a^2 + \pi \alpha^2 a^2 \right) 
  = 1 - \frac{2\pi\alpha a^2}{\tau_1 \tau_2} + O\left( a^4 \right).
\end{align}
\end{proof}
\begin{lemma}\label{appu2}
The following approximation for $u_2$ is valid in the space $H^1$:
\begin{align}\label{uu2}
{\rm If} \quad \widetilde{u}_2 = \left\{
 \begin{array}{ll}
  \dfrac{\ei}{2} \left(\I\bet \cdot \br \right)^2, & r<a, \\[2mm]
  \dfrac{1}{2} \left(\I\bet \cdot \br \right)^2, & r>a,
 \end{array}
 \right. \quad {\rm then} \quad \|u_2-\widetilde{u}_2\|_{H^1}\leq Ca^2.
\end{align}
\end{lemma}
\begin{proof}
Denote a $\sigma$-neighborhood of $\ds$ by $(\ds)_\sigma$. We fix $\sigma$ in such a way that $r>a$ in $(\ds)_\sigma$.
From Lemmas \ref{appu1}, \ref{lambda2} it follows that $u_2$ is a solution of the problem
\begin{align}\nonu
 -\frac{1}{\ep}\, \Delta u_2 &=1+ O\left( a^2 \right), \quad \br \in \s, \quad r\neq a, \\[2mm]
 \left. \left\llbracket u_2 (\br) \right\rrbracket \right . &=0, \quad
\left. \left\llbracket \frac{1}{\ep} \frac{\de u_2 (\br)}{\de n} \right\rrbracket \right. =0, \nonu \\[2mm]
 \rrbracket u_2 (\br) \llbracket &=  \left\rrbracket (\I\bet \cdot \br)u_1 - \frac{1}{2}\,(\I\bet \cdot \br)^2\right\llbracket,
 \nonu
 \end{align}
 where $u_1$ has the following form in $(\ds)_\sigma$:
 \begin{equation}\label{u1h1}
 u_1=(\I\bet \cdot \br)+h_1, \quad \|h_1\|_{H^1((\ds)_\sigma)}\leq Ca^2.
 \end{equation}

 Let $\eta=\eta(\br)\in C^\infty, ~\eta(\br)=1 $ in $(\ds)_{\sigma/2},~\eta(\br)=0$ in $S\backslash (\ds)_\sigma$. Then
 \[
 \|\eta(\br)(\I\bet \cdot \br)h_1\|_{H^1(S)}\leq Ca^2,
\]
 and therefore it is enough  to prove estimate \rf{uu2} for $v_2-\widetilde{u}_2$ where $v_2=u_2-\eta(\br)(\I\bet \cdot \br)h_1$.
 Obviously, $v_2$ satisfies the relations
 \begin{align}\nonu
 -\frac{1}{\ep}\, \Delta v_2 &=1+ f_2, \quad \br \in \s, \quad r\neq a, \\[2mm]
 \left. \left\llbracket v_2 (\br) \right\rrbracket \right . &=0, \quad
\left. \left\llbracket \frac{1}{\ep} \frac{\de v_2 (\br)}{\de n} \right\rrbracket \right. =0, \nonu \\[2mm]
 \rrbracket v_2 (\br) \llbracket &=  \left\rrbracket  \frac{1}{2}\,(\I\bet \cdot \br)^2 \right\llbracket,
 \nonu
 \end{align}
 where $f_2=O\left( a^2 \right)+\Delta[\eta(\br)(\I\bet \cdot \br)h_1]$. The same relations with $f_2=0$ are valid for $\widetilde{u}_2$. Hence,
 Lemma \ref{lemma3} provides the estimate \rf{uu2} for $v_2-\widetilde{u}_2$ if $\|f_2\|_{L_2}\leq Ca^2$. The latter inequality follows from \rf{u1h1}. Indeed,
  $\Delta u_1=\Delta (\I\bet \cdot \br)=0$. Thus $\Delta h_1=0$, and therefore,
  \[
  \Delta[\eta(\br)(\I\bet \cdot \br)h_1]=\Delta[\eta(\br)(\I\bet \cdot \br)]h_1+2\left\langle\nabla[\eta(\br)(\I\bet \cdot \br)],\nabla h_1\right\rangle.
  \]
  This and \rf{u1h1} imply the estimate on $f_2$ and complete the proof of the lemma.
\end{proof}

\subsection{Approximation of $u_3$}
From Lemma \ref{appu1} and \rf{zeta_ser} it follows that $\|u_1-(\I\bet \cdot \br)\|_{L_2}\leq Ca\ln \frac{1}{a}$. Together with Lemma \ref{lambda2}, it allow us to rewrite problem \rf{u_3}, \rf{bc1uk}, \rf{bc2uk}, \rf{u3b} for $u_3$ in the form
\begin{align}
  \label{u3exp}
 -\frac{1}{\ep}\, \Delta u_3  &= (\I\bet \cdot \br)+f, \quad r\neq a, \quad \|f\|_{L_2}\leq Ca\ln \frac{1}{a}
 \\[2mm]
\label{u3a}
 \llbracket u_3 (\br)  \rrbracket  &= 0,
 \quad \left\llbracket \frac{1}{\ep} \frac{\de u_3(\br)}{\de n} \right\rrbracket = 0, \\[2mm]
 \rrbracket u_3(\br) \llbracket &=  \left\rrbracket  (\I\bet \cdot \br)u_2 - \frac{1}{2}\,(\I\bet \cdot \br)^2 u_1
 + \frac{1}{6}\,(\I\bet \cdot \br)^3  \right\llbracket.
 \label{u3con}
\end{align}
Similar to the previous case we formulate
\begin{lemma}
 \label{appu3}
 \begin{align}
 \label{ut3}
 If \quad \ut_3 = \left\{
 \begin{array}{ll}
  \dfrac{\ei}{6} \left(\I\bet \cdot \br \right)^3, & r<a, \\[2mm]
  \dfrac{1}{6} \left(\I\bet \cdot \br \right)^3, & r>a,
 \end{array}
 \right. \quad then \quad \|u_3-\ut_3\|_{H^1}\leq Ca \ln \frac{1}{a}.
\end{align}
\end{lemma}
\begin{proof}
We will use notation $(\ds)_\sigma$ and $\eta (\br)$ from the previous Lemma.
We need to single out the main therm (as $a \to 0$) of the right-hand side of \rf{u3con}.
Lemmas \ref{appu1} and \ref{appu2} imply
\begin{align}\label{u3h1}
 u_1 &= \I\bet \cdot \br + h_1, \quad \|h_1\|_{H^1((\ds)_\sigma)}\leq C_1 a^2, \\[2mm]
 u_2 &= \frac{1}{2}\, (\I\bet \cdot \br)^2 + h_2, \quad \|h_2\|_{H^1((\ds)_\sigma)}\leq C_2 a^2.
 \label{u3h2}
 \end{align}

We introduce function $v_3 = u_3 - g_3, ~g_3 = \eta(\br)\left((\I\bet \cdot \br)h_2
- \frac{1}{2}\,(\I\bet \cdot \br)^2 h_1 \right)$. This function satisfies the relations
\begin{align}\nonu
 -\frac{1}{\ep}\, \Delta v_3 &= \I\bet \cdot \br+ f_3, \quad \br \in \s, \quad r\neq a, \\[2mm]
 \left. \left\llbracket v_3 (\br) \right\rrbracket \right . &=0, \quad
\left. \left\llbracket \frac{1}{\ep} \frac{\de v_3 (\br)}{\de n} \right\rrbracket \right. =0, \nonu \\[2mm]
 \rrbracket v_3 (\br) \llbracket &=  \left\rrbracket  \frac{1}{6}\,(\I\bet \cdot \br)^3\right\llbracket,
 \nonumber
 \end{align}
 where $f_3=O\left( a\ln a \right)+ \Delta g_3,$ and the last relation above (for the jump of $v_3$ on $\ds$) is $a$-independent.
 From \rf{u3h1},\rf{u3h2} it follow that $\|g_3\|_{H^1} \leq Ca^2$, and therefore one can prove Lemma \ref{appu3} for $v_3$ instead of $u_3$.
 We note that $\ut_3$ satisfies the same relations as those for $v_3$ with $f_3=0$. Thus,  estimate \rf{ut3} for $v_3 - \ut_3$ will follow
 from Lemma \ref{lemma3} if we show that $\|f_3\|_{L_2} \leq C a\ln \frac{1}{a}$. Thus, to complete the proof of the Lemma it suffices to show that
 $\|\Delta g_3\|_{L_2} \leq Ca^2$.

 From equations \rf{u_1},\rf{u_2} (where $\lambda_1 = 0$) and \rf{u3h1}, \rf{u3h2} it follows that
 $\Delta h_1 =0$,  $\Delta h_2 = 1 - \lambda_2 = O\left( a^2 \right)$. Hence
\begin{align}
 &\Delta g_3 =
 h_2 \,\Delta\left(\eta(\br)\left(\I\bet \cdot \br \right)\right)
 +2\left\langle\nabla[\eta(\br)(\I\bet \cdot \br)],\nabla h_2\right\rangle
 + \eta(\br)\left(\I\bet \cdot \br \right)O(a^2) \nonu \\[2mm]
 &-\frac{1}{2}\,h_1  \,\Delta\left(\eta(\br)\left(\I\bet \cdot \br \right)^2\right)
-\left\langle\nabla[\eta(\br)(\I\bet \cdot \br)^2],\nabla h_1\right\rangle.
\end{align}
Now the desired estimate on $\Delta g_3$ follows from \rf{u3h1},\rf{u3h2}.
\end{proof}

\section{Effective dielectric tensor and the dispersion relation}
\setcounter{equation}{0}

Let us recall that the two-dimensional electric component of the TE-mode is determined by
$\BE = \dfrac{\I c}{\om \ep}\, \nabla \times \BH = \dfrac{\I c}{\om \ep} \,[u_y, -u_x]$ and that $\BD = \ep \BE.$
The effective dielectric tensor $\bep^\ast$ of the problem relates the average electric field
and the average electric displacement over the fundamental cell
\begin{equation}
 \BDA = \bep^\ast \BEA.
\end{equation}
In the principal axes $ \bep^\ast$ has a diagonal form and can be found from the relation
\begin{equation}
\int_{\s} \np u\, \D S = \bep^\ast \int_{\s} \frac{1}{\ep} \,\np u\, \D S, \quad \np u = [u_y, -u_x].
\label{bep_def}
\end{equation}
We represent $\bep^\ast$ in the form
\begin{equation}
 \bep^\ast = \BI + \bepo,
 \label{bepo}
\end{equation}
where $\BI$ is a $2\times 2$ identity matrix and $
 {\bepo} = \left[
 \begin{array}{cc}
  \epo_1 & 0 \\
  0  & \epo_2
 \end{array}
 \right].
$
Substituting \rf{bepo} into \rf{bep_def} we obtain equation for $\bepo$
\begin{equation}
 \bepo \int_{\s} \frac{1}{\ep} \,\np u\, \D S =  \left(1 - \frac{1}{\ei}\right) \int_{\si}  \np u\, \D S.
 \label{bepo1}
\end{equation}
Observe that in the right-hand side of \rf{bepo1} integration is performed only over $\si$.
With expansion \rf{u_ser}, Theorem \ref{theorem1} and Theorem \ref{theorem2} on the uniform convergence of $u$ (see Appendix A) we obtain for the entries of $\bepo$
\begin{align}
 \label{epo1}
 \epo_1 \int_{\s} \dfrac{1}{\ep}\, \de_y \left( u_1 + q^2 u_3 \right) \D S  &= \frac{\ei - 1}{\ei}\, \int_{\si} \de_y \left(
 \ui_1 + q^2 \ui_3 \right) \D S +  O\left ( \left(q^2 + a^2 \right)^{\frac{5}{2}}\right), \\[2mm]
 \label{epo2}
 \epo_2 \int_{\s} \dfrac{1}{\ep}\, \de_x \left( u_1 + q^2 u_3 \right) \D S  &= \frac{\ei - 1}{\ei}\, \int_{\si} \de_x \left(
 \ui_1 + q^2 \ui_3 \right) \D S +  O\left ( \left(q^2 + a^2 \right)^{\frac{5}{2}}\right),
\end{align}
Now using \rf{A0}-\rf{D0},  we evaluate the integrals involved in \rf{epo1}-\rf{epo2}
\begin{align}
 \int_{\si}  \de_x \ui_1\, \D S &= \pi \I a^2 A_1 = \frac{2\pi \I \ei a^2}{\ei + 1}\left( 1 - \frac{2\alpha \eta_1}{\tau_1}\,a^2 \right) \cos \te
 + O\left (  a^6 \right), \\[2mm]
 \int_{\si}  \de_y \ui_1\, \D S &= \pi \I a^2 B_1 = \frac{2\pi \I \ei a^2}{\ei + 1}\left( 1 - \frac{2\alpha \etat}{\tau_2}\,a^2 \right) \sin \te
 + O\left (  a^6 \right), \\[2mm]
 \int_{\si}  \de_x \ui_3\, \D S &= -\frac{\I \ei}{2} \cos \te \int_{\si} (x\cos \te + y\sin \te)^2 \, \D S = -\frac{\pi \I \ei a^4}{8} \cos \te
 + O\left (  a^5 \ln \frac{1}{a} \right), \\[2mm]
 \int_{\si}  \de_y \ui_3\, \D S &= -\frac{\I \ei}{2} \sin \te \int_{\si} (x\cos \te + y\sin \te)^2 \, \D S = -\frac{\pi \I \ei a^4}{8} \sin \te
 + O\left (  a^5 \ln \frac{1}{a} \right).
\end{align}
Integrals over $\se = \s \smallsetminus \si$ are evaluated using Green's theorem
\begin{align}
 &\int_{\se} \np \ue_1\, \D S
 = -\oint_{\ds} \left[ \ue_{1} \,\D x, \ue_{1}\,\D y\right]
 + \oint_{\dsi} \left[ \ui_{1} \,\D x, \ui_{1}\,\D y\right].
\end{align}
Integrals over the boundary $\ds$ are evaluated by the property \rf{u1b} of $u_1$
\begin{align}
 \oint_{\ds} \ue_{1} \,\D x &= \int_{-\tau_1/2}^{\tau_1/2} \ue_1 \left(x, -\frac{\tau_2}{2} \right) \D x -
 \int_{-\tau_1/2}^{\tau_1/2} \ue_1 \left(x, \frac{\tau_2}{2} \right) \D x = -\I \tau_1 \tau_2 \sin \te, \\[2mm]
 \oint_{\ds} \ue_{1} \,\D y &= \int_{-\tau_2/2}^{\tau_2/2} \ue_1 \left(\frac{\tau_1}{2},y \right) \D y -
 \int_{-\tau_2/2}^{\tau_2/2} \ue_1 \left(-\frac{\tau_1}{2},y \right) \D y = \I \tau_1 \tau_2 \cos \te.
\end{align}
Second integral is evaluated using \rf{u1_in_appr}, \rf{A0}, and \rf{D0}
\begin{align}
 \oint_{\dsi} \ui_{1} \,\D x &= -\I a^2\int_0^{2\pi} \left(A_1 \cos \phi + B_1 \sin \phi \right) \sin \phi \,\D \phi
 = -\pi \I a^2 B_1  \nonu \\[2mm]
 &= -\frac{2\pi \I a^2 \ei}{\ei + 1} \sin \te \left(1 - \frac{2\alpha a^2 \etat}{\tau_2} + O\left(a^6\right)\right).
 \end{align}
Similarly,
 \begin{align}
 \oint_{\dsi} \ui_{1} \,\D y &=
 \frac{2\pi \I a^2 \ei}{\ei + 1} \cos \te \left(1 - \frac{2\alpha a^2 \eta_1}{\tau_1}
 +  O\left(a^6\right)\right).
\end{align}
Finally from \rf{ut3} we estimate integrals of $u_3$ to the order $O\left( a^4 \right)$
\begin{align}
 \int_{\s} \de_x u_3 \, \D S &= -\frac{\I}{2}\, \cos \te \int_{\s} (x\cos \te + y \sin \te)^2\, \D x \D y
 = -\frac{\I \tau_1 \tau_2}{24}\, \cos \te \left(\tau_1^2 \cos^2 \te + \tau_2^2 \sin^2 \te \right), \\[2mm]
 \int_{\s} \de_y u_3 \, \D S &= -\frac{\I}{2}\, \sin \te \int_{\s} (x\cos \te + y \sin \te)^2\, \D x \D y
 = -\frac{\I \tau_1 \tau_2}{24}\, \sin \te \left(\tau_1^2 \cos^2 \te + \tau_2^2 \sin^2 \te \right).
\end{align}
Substituting evaluated integrals into \rf{epo1}-\rf{epo2} we obtain components of the effective tensor with the accuracy
$O\left( \left(q^2 + a^2\right)^{\frac{5}{2}} \right)$
\begin{align}
 \ei_1^\ast &= 1 + \frac{2\pi \alpha a^2 \left(1 - \frac{2\alpha a^2 \etat}{\tau_2}\right)}
 {\tau_1 \tau_2 - 2\pi \alpha a^2 \left(1 - \frac{2\alpha a^2 \etat}{\tau_2}\right)} +\frac{q^2}{12}\,
 \frac{\pi \alpha a^2 \tau_1 \tau_2 \left(1 - \frac{2\alpha a^2 \etat}{\tau_2}\right)\left(\tau_1^2 \cos^2 \te + \tau_2^2 \sin^2 \te \right)}
 {\left(\tau_1 \tau_2 - 2\pi \alpha a^2 \left(1 - \frac{2\alpha a^2 \etat}{\tau_2}\right)\right)^2}, \\[2mm]
 \ei_2^\ast &= 1 + \frac{2\pi \alpha a^2 \left(1 - \frac{2\alpha a^2 \eta_1}{\tau_1}\right)}
 {\tau_1 \tau_2 - 2\pi \alpha a^2 \left(1 - \frac{2\alpha a^2 \eta_1}{\tau_1}\right)} +\frac{q^2}{12}\,
 \frac{\pi \alpha a^2 \tau_1 \tau_2 \left(1 - \frac{2\alpha a^2 \eta_1}{\tau_1}\right)\left(\tau_1^2 \cos^2 \te + \tau_2^2 \sin^2 \te \right)}
 {\left(\tau_1 \tau_2 - 2\pi \alpha a^2 \left(1 - \frac{2\alpha a^2 \eta_1}{\tau_1}\right)\right)^2}.
\end{align}
With Legendre's relation  $\eta_1 \tau_2 + \etat \tau_1 = \pi$
the series expansion of $\bep^\ast$ gives
\begin{align}
 \bep^\ast &= \left(1 + \frac{2\pi \alpha a^2}{\tau_1 \tau_2}
 +\frac{1}{12}\, \frac{\pi \alpha a^2 q^2}{\tau_1 \tau_2} \left( \tau_1^2 \cos^2 \te + \tau_2^2 \sin^2 \te \right) \right) \BI \nonu \\[2mm]
 &+ \frac{4\pi \alpha^2 a^4}{\tau_1^2 \tau_2^2 } \left[
\begin{array}{cc}
 \eta_1 \tau_2 & 0 \\[1mm]
 0 & \etat \tau_1
\end{array}
\right]+ O\left ( \left(q^2 + a^2\right)^{\frac{5}{2}} \right).
\end{align}

For the square lattice $\tau_1 = \tau_2 = \tau$ we have $\eta_1 = \etat = \dfrac{\pi}{2\tau}$ \cite{Abramowitz:1964}.
Then in our approximation $\bep^\ast$ becomes isotropic $\bep^\ast = \ei^\ast \BI$, where
\begin{equation}
\label{eps}
 \ei^\ast = 1 + 2\alpha \,\frac{\pi a^2}{\tau^2} + 2\alpha^2 \left(\frac{\pi a^2}{\tau^2} \right)^2
 + \frac{1}{12}\,\pi \alpha a^2 q^2 + O\left ( \left(q^2 + a^2\right)^{\frac{5}{2}} \right).
\end{equation}
 It should be remarked that while the static part of \rf{eps} agrees with the expansion of Maxwell's
 formula the frequency-dependent correction differs substantially from that obtained in \cite{McPhedran:1996}.

Going back to evaluation of \rf{nu2} one can easily check that
\begin{align}
 |\nabla u_2|^2 + 2\re\left( \nabla u_1 \cdot \nabla \bar{u}_3 \right)&= O \left(a^2 \right), \\[2mm]
 |u_1|^2 + 2\re u_2 &= O\left( a^2 \right).
\end{align}
Then from \rf{lambda_2} and \rf{nu2} we obtain the dispersion relation
\begin{align}
\nu^2 = q^2 \left(1 - \frac{2\alpha \pi a^2}{\tau_1 \tau_2} \right) + O\left( q^2 (q^2 + a^2)^2 \right).
\label{dr}
\end{align}
Thus the first correction in the dispersion relation does not depend of the shape of the lattice but only on
the concentration of the scatterers.

\section{Conclusion}
We have considered the problem of transverse propagation of electromagnetic waves through a doubly periodic
rectangular array of circular dielectric cylinders of radius $a$. Solution of the problem is sought in the form of a power
series in terms of the magnitude $q$ of the quasimomentum of the Bloch wave. We prove that the eigenfunction and the eigenvalue are
analytic functions of $q^2$ that converge uniformly in $a$. We find explicitly frequency correction terms to the effective
dielectric tensor as well as to the dispersion relation and rigorously estimate the remainders.
The approach devised in the paper can also be used to find higher order terms of the effective tensor
and the dispersion relation.

%
%
%
%
%
%



\section*{Appendix A. Uniform property of the series expansion}
\setcounter{equation}{0}
\setcounter{section}{1}
\renewcommand{\theequation}{\Alph{section}.\arabic{equation}}

The following theorem shows that the series expansion of the eigenfunction $u$ is uniform in $a$.
\begin{theorem}
 \label{theorem2}
 Let $a \leq a_0 \leq \min (\tau_1, \tau_2)$. Then there are constants $q_0, \lambda_0 > 0$ such that the eigenvalue problem \rf{Hz1}-\rf{bc1a}
 has
 a unique eigenvalue $\lambda = \nu^2$ when $q \leq q_0, |\lambda| \leq \lambda_0$, and the eigenvalue is simple.
 The corresponding eigenfunction $u=u(\br, \Be, a)$ normalized by the condition
 \begin{equation}
 \label{norm}
   \int_{S}u(\br, \Be)\,dS=|\s|=\tau_1\tau_2
 \end{equation}
 is analytic in $q$
 \begin{align}
  u = 1 + \sum_{n=1}^\infty u_n (\br, a)\, q^n, \quad q \leq q_0, \quad \|u_n\|_{H^1} \leq C_n,
 \end{align}
 where $C_n$ do not depend on $a$ and the series converges in $H^1 (\s)$ uniformly in $a$. The corresponding eigenvalue
 $\lambda$ can also be expanded in a power series in $q, q \leq q_0$, which converges uniformly in $a, a \leq a_0$.
\end{theorem}
\begin{proof}
 Let us reduce the problem \rf{Hz1}-\rf{bc1a} to an equivalent one where the domain of the operator does not depend on $q$.
 Let
 \begin{equation}
  \beta (\br) = 1 + \left(\E^{\I \Be \cdot \br} -1 \right) \alpha (\br),
 \end{equation}
 where $\alpha (\br)$ is a $C^{\infty}$ function whose graph is shown in Figure 2. The substitution $u = \beta v$ in \rf{Hz1}-\rf{bc1a}
 and multiplication of the equation by $\beta^{-1}$ reduces the problem to the following one

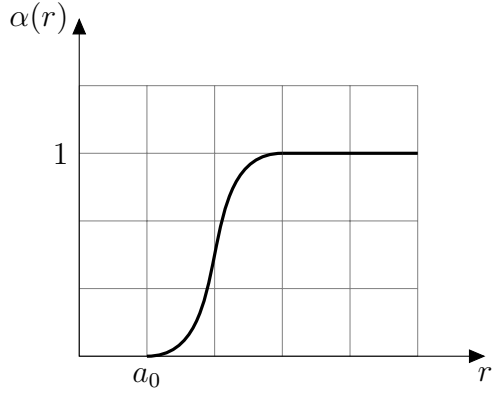
\begin{figure}[H]
\begin{center}
\begin{tikzpicture}[>=triangle 45,scale=0.90]

\draw [help lines] (0,0) grid (5,4);
\draw[<->] (6,0) node[below]{$r$} -- (0,0) --
(0,5) node[left]{$\alpha(r)$} ;

\draw[very thick] (1,0) node[below] {$a_0$} to [out=0,in=180] (3,3);
\draw[very thick] (3,3) -- (5,3);
\node [left] at (0,3) {$1$};

\end{tikzpicture}

\caption{Graph of the function $\alpha (r)$.}
\end{center}
\label{fig:2}
\end{figure}

 \begin{align}
  \left(-\frac{1}{\ep}\, \Delta + B_q \right) v = \lambda v, \quad
  B_q v = \frac{1}{\ep \beta} \left( 2\nabla \beta \cdot \nabla v + v \Delta \beta \right),
  \label{Bqv}
 \end{align}
 where $\lambda = \nu^2$ and the domain $\mathscr{D} = \mathscr{D}(a)$ of operators $\dst \frac{1}{\ep}\, \Delta $ and $B_q$ consists of functions
 $v \in H^2 \left( \si \right) \oplus  H^2 \left( \se \right)$ that satisfy $
 \left\llbracket v (\br) \right\rrbracket  =0,  \left\llbracket \dfrac{1}{\ep} \dfrac{\de v (\br)}{\de n} \right\rrbracket  =0$  and the periodicity
 condition $\rrbracket  v(\br) \llbracket =0$.

We will need the following lemma:
\begin{lemma}
 There exist constants $\gamma_1, \gamma_2 >0$ such that the operator $ \dst \frac{1}{\ep}\, \Delta  + B_q$ does not have eigenvalues
  in the annulus $\gamma_1 q \leq | \lambda| \leq \gamma_2$ when $a \leq a_0$, $0 < q \leq \gamma_2/ \gamma_1$
\end{lemma}
\begin{proof}
 We normalize the eigenfunction $v$ in \rf{v} by the condition $\| v\|_{L^2} = 1$. Thus
 \begin{align}
  \label{star}
  \sum_{m,n=0}^\infty |v_{nm}|^2 = \frac{1}{\tau_1 \tau_2}, \quad \text{where} \quad
  v = \sum_{m,n=0}^\infty v_{mn} \E ^{2\pi \I \left( \frac{mx}{\tau_1} + \frac{ny}{\tau_2} \right)}.
 \end{align}
Let us show that the coefficient $v_{00}$ cannot be very small. Clearly, $\dst \dfrac{1}{\ep \beta} \left(2|\nabla \beta | + |\Delta \beta| \right) \leq C q$
for small $q$. Thus,
\begin{align}
 \label{star2}
 \| B_q v \|_{L_2} \leq Cq \left(\| \nabla v \|_{L_2} + 1 \right).
\end{align}
From here and Green's formula applied to \rf{v} it follows that
\begin{align}
  \int_{\s} \frac{1}{\ep}\, \left| \nabla v \right|^2 \D S \leq
  \| B_q v \|_{L_2} \cdot \| v\|_{L_2} + |\lambda | \|v\|^2_{L_2} \leq
  Cq \left( \| \nabla v \|_{L_2} + 1 \right) + |\lambda|, \quad q \ll 1.
 \end{align}
 Hence,
 \begin{align}
 \label{star3}
  \| \nabla v\|^2_{L_2} \leq C_1 \left( |\lambda| + q \right) ~~ \text{for} ~~ q \ll 1,
 \end{align}
 i. e.
 \begin{align}
  \left( 2\pi \right)^2 \sum_{(m,n)\neq (0,0)} |v_{m,n} |^2
  \left[ \left( \frac{m}{\tau_1}\right)^2 + \left( \frac{n}{\tau_2}\right)^2\right] \tau_1 \tau_2
  \leq C_1 \left( |\lambda| + q \right) ~~ \text{for} ~~ q \ll 1,
 \end{align}
This and \rf{star} imply the existence of $C_0 > 0$ such that $|v_{00}| > C_0$ for small enough
$|\lambda| + q$.

Lemma \ref{lemma1} implies that a non-trivial solution of \rf{Bqv} exists only if
 \begin{equation}
  \int_{\s} \left( \lambda v - B_q v\right) \D S =0.
  \label{Bq}
 \end{equation}
 From \rf{star2} and \rf{star3} it follows that
 \begin{align}
  \int_{\s} |B_q v| \,\D S \leq C_2 \left[ q \sqrt{|\lambda| +q} +q \right] \leq C_3 q
 \end{align}
if $|\lambda| +q$ is small. Thus \rf{Bq} implies that $|\lambda | C_0 \leq C_3 q$ for small eigenvalues when $q \ll 1$, i.e.,
there exists $\gamma_2>0$ such that eigenvalues $\lambda$ in the circle $|\lambda|\leq\gamma_2$ are located only inside
of a smaller circle $|\lambda|<\gamma_1 q, \gamma_1=C_3/C_0,$ when $q$ is small enough.
\end{proof}

Continuing the proof of the theorem we assume below that $a \leq a_0$, $q \leq \gamma_2/\gamma_1$.
Thus, the circle
$\Gamma = \{\lambda : |\lambda| = \gamma_2 \}$ splits the spectrum of $\dst \frac{1}{\ep}\, \Delta + B_q$
into two parts. Since operator
$\dst \frac{1}{\ep}\, \Delta + B_q: \mathscr{D}(a) \to L_2$, where $\mathscr{D}(a)$ was
defined in \rf{Bqv}, has a discrete spectrum, operator
\begin{equation}
 P_q = \int_{\Gamma} \left( \frac{1}{\ep}\, \Delta + B_q -\lambda \right)^{-1} \D \lambda
\end{equation}
is a projection on the space spanned by the eigenfunctions of $\dst \frac{1}{\ep}\, \Delta + B_q$ with eigenvalues inside $\Gamma$. We will show below
that $\| P_q - P_0 \| < 1$ if $q$ and $\gamma_2$ are small enough. Hence \cite[sec.~XII.2]{RS4:1978},  the ranges of $P_q$ and $P_0$ have the same dimensions.
We reduce $\gamma_2$, if needed, to guarantee that $P_0$ is the projection on the simple eigenfunction $u=const$ of $\dst \frac{1}{\ep}\,\Delta$.
Then
$\dst \frac{1}{\ep}\, \Delta + B_q$ has a unique simple eigenvalue in $\Gamma$ when $q, \gamma_2 \ll 1$. The corresponding eigenfunction
is proportional to $P_q f$ with an arbitrary $f$ such that $P_q f \neq 0$. Function $f$ needs to be normalized to guarantee \rf{norm}.

It was shown above that $\dst \frac{1}{\ep}\, \Delta - \lambda$ is invertible when $\lambda \in \Gamma$,
i. e. $\dst \left( \frac{1}{\ep}\, \Delta - \lambda \right)^{-1}: L_2 \to H^1$ is bounded.
We need an estimate for this operator with a constant that does not depend on $a$.

Let
\begin{align}
 \left( \frac{1}{\ep}\, \Delta - \lambda \right)u = f \in L_2, \quad u \in \mathscr{D}(a), \quad a \leq a_0, \quad \lambda \in \Gamma.
\end{align}
Lemma \ref{lemma1} implies that
\begin{align}
 \int_{\s} \left( \lambda u + f\right) \D S =0,
 \label{orth}
\end{align}
and from Green's formula it follows that
\begin{align}
 \left \| \frac{1}{\ep}\, \nabla u \right \|^2_{L_2} - |\lambda | \|u \|^2_{L_2} \leq
 \int_{\s} |fu|\, \D S \leq \gamma_2 \|u \|^2_{L_2} + \frac{1}{\gamma_2} \|f \|^2_{L_2}.
\end{align}
Thus,
\begin{align}
\|\nabla u \|^2_{L_2} \leq \left( 2\gamma_2 \|u \|^2_{L_2} +\frac{1}{\gamma_2}\, \|f \|^2_{L_2} \right) \max \left(\ep \right).
\end{align}
Hence,
\begin{align}
  \sum_{(m,n)\neq (0,0)} |u_{m,n} |^2  \left[ \left( \frac{m}{\tau_1}\right)^2 + \left( \frac{n}{\tau_2}\right)^2\right]
  \leq C \gamma_2 \left[ 2\gamma_2 \sum_{(m,n)\neq (0,0)} |u_{m,n} |^2 + |u_{0,0}|^2 \right] +\frac{1}{\gamma_2}\, \| f\|^2_{L_2},
 \end{align}
where $u_{m,n}$ are Fourier coefficients of $u$. If $\gamma_2$ is small enough then the latter estimate and \rf{orth} imply
\begin{align}
  \sum_{(m,n)\neq (0,0)} |u_{m,n} |^2  \left[ \left( \frac{m}{\tau_1}\right)^2 + \left( \frac{n}{\tau_2}\right)^2\right]
  \leq C_1 \frac{|f_{00}|^2}{\gamma_2} + \frac{1}{\gamma_2}\,\|f \|^2_{L_2} \leq C_2 \| f\|^2_{L_2},
 \end{align}
From here and \rf{orth} it also follows that
\begin{align}
 \| u\|^2_{L_2} = \left[  \sum_{(m,n)\neq (0,0)} |u_{m,n} |^2  + |u_{0,0}|^2 \right] \tau_1 \tau_2 \leq C \| f\|^2.
\end{align}
Thus, $\left\| \left( \dfrac{1}{\ep}\, \Delta - \lambda \right)^{-1} f \right\|_{H^1} \leq C \| f \|_{L_2}$, where $C$ does not depend on $a$ and $\lambda \in \Gamma$.

Now we write
\begin{align}
  \frac{1}{\ep}\, \Delta + B_q - \lambda = (1 +T_q) \left( \frac{1}{\ep}\, \Delta - \lambda \right), \quad
  T_q = B_q \left( \frac{1}{\ep}\, \Delta - \lambda \right)^{-1}: L_2 \to L_2,
\end{align}
where operator $T_q$ is analytic in $q$, its power series converges in the norm space uniformly in $a$ and $\|T_q\| \to 0$ as $|q| \to 0$.
It remains to write $P_q$ in the form
\begin{equation}
 P_q = \int_{\Gamma} \left( \frac{1}{\ep}\, \Delta - \lambda \right)^{-1} \sum_{n=0}^\infty \left( -T_q \right)^n \D \lambda
\end{equation}
and expand $T_q$ in a power series in $q$.
This proves that $\| P_q - P_0 \| < 1$ if $q \ll 1$ and provides a power series for $P_q f$
which converges in $H^1 (\s)$ uniformly in $a$. Power expansion of $\lambda = \nu^2$ follows immediately from \rf{green}.

\end{proof}

%

\end{document}